\documentclass{lmcs}
\pdfoutput=1

\usepackage{lastpage}
\lmcsdoi{16}{4}{11}
\lmcsheading{}{\pageref{LastPage}}{}{}%
{Oct.~11,~2019}{Nov.~30,~2020}{}

\keywords{inductive theorem proving, arithmetical theories, clause logic, cyclic proofs}

\usepackage[utf8]{inputenc}
\usepackage[T1]{fontenc}
\usepackage{graphicx}
\usepackage{grffile}
\usepackage{longtable}
\usepackage{wrapfig}
\usepackage{rotating}
\usepackage[normalem]{ulem}
\usepackage{amsmath}
\usepackage{textcomp}
\usepackage{amssymb}
\usepackage{capt-of}
\usepackage{url}

\usepackage[utf8]{inputenc}
\usepackage{amsmath}
\usepackage{amssymb}
\usepackage{amsthm}
\usepackage{stmaryrd}
\usepackage{todonotes}
\usepackage{bussproofs}
\usepackage{pdflscape}
\usepackage{hyperref}
\usepackage{mathtools}
\usepackage[nolist, nohyperlinks]{acronym}

\newcommand{\numeral}[1]{\overline{#1}}

\newcommand{\Succ}{\csym{s}}
\newcommand{\Param}{\eta}
\newcommand{\Zero}{\csym{0}}
\newcommand{\csym}[1]{\mathsf{#1}}
\newcommand{\Pred}{\mathsf{p}}
\newcommand{\TriSymbol}{\triangleright}
\newcommand{\Tri}[2]{{#1} \, \triangleright \, {#2}}

\newcommand{\PredP}{\mathsf{P}}
\newcommand{\PredQ}{\mathsf{Q}}

\newcommand{\fsort}[1]{\mathsf{#1}}
\newcommand{\nat}{\fsort{nat}}
\newcommand{\sortNat}{\nat}

\newcommand{\TDelta}{T_\triangleright}
\newcommand{\TDeltaSimple}{T_\triangleright'}
\newcommand{\TDeltaI}{T_\triangleright^{I}}

\newcommand{\SuccInter}{\mathbf{S}}
\newcommand{\PlusInter}{\boldsymbol{+}}
\newcommand{\ZeroInter}{\mathbf{0}}
\newcommand{\PredInter}{\mathbf{P}}

\newcommand{\TriInterSymbol}{\blacktriangleright}
\newcommand{\TriInter}[2]{{#1} \ \blacktriangleright {#2}}

\newcommand{\theoryIndE}[2]{I\exists_{#2}(#1)}
\newcommand{\theoryIndEL}[1]{I\exists_{#1}}
\newcommand{\theoryIndOpen}[1]{I\mathrm{Open}(#1)}
\newcommand{\theoryIndOpenL}{I\mathrm{Open}}
\newcommand{\E}[1]{\exists_{#1}}
\newcommand{\A}[1]{\forall_{#1}}

\newcommand{\formulaInd}[2]{I_{#1}{#2}}

\newcommand{\lkid}{\mathsf{LKID}}

\begin{document}

\title{Clause Set Cycles and Induction}

\author[S.~Hetzl]{Stefan Hetzl}
\address{Vienna University of Technology\protect\\Institute of Discrete Mathematics and Geometry}	
\email{\{stefan.hetzl,jannik.vierling\}@tuwien.ac.at}

\author[J.~Vierling]{Jannik Vierling}
\thanks{Supported by the Vienna Science and Technology Fund (WWTF) project VRG12-004}

\begin{abstract}
  In this article we relate a family of methods for automated inductive theorem proving based on cycle detection in saturation-based provers to well-known theories of induction.
  To this end we introduce the notion of clause set cycles---a formalism abstracting a certain type of cyclic dependency between clause sets.
  We first show that the formalism of clause set cycles is contained in the theory of \(\E{1}\) induction.
    Secondly, we consider the relation between clause set cycles and the theory of open induction.
  By providing a finite axiomatization of a theory of triangular numbers with open induction we show that the formalism of clause set cycles is not contained in the theory of open induction.
  Furthermore, we conjecture that open induction and clause set cycles are incomparable.
  Finally, we transfer these results to a concrete method of automated inductive theorem proving called the n-clause calculus.
\end{abstract}

\maketitle

\section{Introduction}

\label{sec:org62947a1}
The subject of \ac{AITP} aims at automating the process of finding proofs by mathematical induction.
\ac{AITP} is of paramount importance to the formal verification of software and hardware.
Every non-trivial program contains loops or recursion, hence its verification requires some inductive reasoning.
But also the development of proof assistants can benefit from automated inductive theorem proving by providing hammers that can discharge lemmas automatically.

It is folklore that finding proofs by induction is difficult because of the necessity of non-analyticity of induction formulas.
The non-analyticity of induction formulas can be explained proof-theoretically by the failure of cut-elimination in LK with an induction rule, see \cite{wong2017some} for a precise statement.
A wide variety of approaches have been proposed to address this problem.
Among others there are approaches based on enhancements of saturation-based provers \cite{cruanes2017superposition, echenim2019combining, kersani2013combining}, cyclic proofs \cite{brotherston2012generic}, rippling \cite{bundy1993rippling}, theory exploration \cite{claessen2013automating}, etc.
Most of these approaches are rather technical in nature and are thus difficult to analyze formally.
Hence, most of the analyses carried out for methods of automated inductive theorem proving are empirical and little is known about the theoretical properties of these methods.
We believe that providing formal analyses of these methods will contribute to the theoretical foundations of the subject and thus help in developing better methods.

The work presented in this article continues the analysis of Kersani and Peltier's n-clause calculus started in the second author's master's thesis \cite{vierling2018cyclic}.
In \cite{vierling2018cyclic} refutations of the n-clause calculus are translated into proofs of the sequent calculus \(\lkid\) introduced in \cite{brotherston2010sequent}.
This translation makes it possible to read off the induction formulas used by the n-clause calculus.
The analysis carried out in \cite{vierling2018cyclic} operates directly on the n-clause calculus as originally defined in \cite{kersani2013combining} without introducing an intermediary abstraction.
Therefore, the results obtained in \cite{vierling2018cyclic} are complicated by inessential technical details.
Furthermore, the technicalities of the n-clause calculus made it difficult to state conjectures clearly.
For instance, the n-clause calculus imposes restrictions on the types of function symbols.
Because of these restrictions it is already difficult to express simple properties such as for example the associativity of natural numbers.
In this article we extend the previous work by introducing an abstraction called clause set cycles.
This new formalism abstracts the inessential details of the n-clause calculus, thus, allowing us to carry out a more systematic analysis and to formulate more general conjectures.

The article is structured as follows.
In Section \ref{sec:org353db4f} we will define the notion of clause set cycle and the associated notion of refutability by a clause set cycle.
We will then situate these notions with respect to \(\E{1}\) induction, and hence show that the formalism is inherently weak.
This result is a generalization of the main theorem (Theorem 6.27) of \cite{vierling2018cyclic}.
In Section \ref{sec:org54c1dd2} we will provide a finite axiomatization for a theory of triangular numbers with open induction.
This result will then serve as the main technical result in Section \ref{sec:org567ac08}, where we will show that the notion of refutability by a clause set cycle is not weaker than open induction.
In Section \ref{sec:orga779b6c} we will show that the n-clause calculus is indeed a special case of the system of clause set cycles and transfer the main results of sections \ref{sec:org353db4f} and \ref{sec:org567ac08} to the n-clause calculus.
We thus answer positively the conjecture of \cite{vierling2018cyclic} that there exists a clause set that is refutable by the n-clause calculus, but that is not refutable with open induction.
As a result we situate the power of the n-clause calculus with respect to the theories of \(\E{1}\) induction and open induction.

\section{\texorpdfstring{Clause Set Cycles and \(\E{1}\) Induction}{Clause Set Cycles and E1 Induction}}
\label{sec:org353db4f}
We work in a many-sorted first-order classical logic.
By a many-sorted first-order classical logic we understand a classical first-order logic with a finite set of sorts \(\mathcal{S}\).
The sorts represent universes and are interpreted as pairwise distinct sets.
Each function symbol \(f\) has a type of the form \(s_{1} \times \dots \times s_{n} \to s_{n+1}\), where \(s_{1}, \dots, s_{n}, s_{n + 1} \in \mathcal{S}\) are sorts and \(n \geq 0\).
Analogously, a predicate symbol \(P\) has a type of the form \(s_{1} \times \dots \times s_{n}\), again with \(s_{1}, \dots, s_{n} \in \mathcal{S}\) and \(n \geq 0\).
Let the sort \(s \in \mathcal{S}\) be interpreted as the set \(U_{s}\), then a function symbol \(f\) with type as above, will be interpreted as a function that maps elements of \(U_{s_{1}} \times \dots \times U_{s_{n}}\) to elements of \( U_{s_{n + 1}}\).
Analogously, a predicate symbol \(P\) with type as above is interpreted as a subset of \(U_{s_{1}} \times \dots \times U_{s_{n}}\).
Each individual variable ranges over a fixed sort.
Whenever the sort of a variable is clear from the context, we will not mention it explicitly.
Terms and formulas are defined as usual, but function symbols and relation symbols need to be applied to terms that agree with the type of the symbol.

Every language that we consider is supposed to contain at least the sort \(\sortNat\) representing the natural numbers, with its function symbols \(\Zero : \sortNat\) representing the number \(0\) and \(\Succ : \sortNat \to \sortNat\) representing the successor function.
In case the language contains only one sort, we specify function symbols by pairs of the form \(f/n\) where \(f\) is a function symbol and \(n\) is a natural number representing the arity of the symbol \(f\).
In the following we fix one such language and denote it by \(L\).
Formulas, structures, models, truth, validity, \(\models\), \(\vdash\), etc. are defined as usual.

By \(\Param\) we denote a distinguished variable ranging over the sort \(\sortNat\).
We will call this variable \(\Param\) a parameter.
The parameter \(\Param\) is mainly used to indicate positions on which arguments by induction take place, that is, the parameter usually plays the role of the induction variable.
Usually the parameter \(\Param\) will occur freely, in other words, it will not be bound by quantifiers and therefore behaves similarly to a constant.
In the literature a similar concept of parameter is used, with the difference that the parameter is usually a constant (even a Skolem constant).
In our case treating the parameter as a variable seems to be more natural, especially when dealing with the language of induction formulas.

Let \(f\) be a unary function symbol and \(t\) a term.
In order to save parentheses we write \(ft\) for the term \(f(t)\).
By \(f^{n}t\) we abbreviate the term
\[
  \underbrace{f(f(\dots f}_{\text{\(n\) times}}(t) \dots )).
\]
Let \(n \in \mathbb{N}\), then by \(\numeral{n}\) we denote the term \(\Succ^{n}\Zero\).
Let \(t, t_{1}, \dots, t_{n}\) be terms of sort \(\sortNat\) and \(+\) be a function symbol of type \(\sortNat \times \sortNat \to \sortNat\) denoting the addition of natural numbers.
For the sake of readability we will use the symbol \(+\) as an infix symbol.
The expression \(t_{1} + t_{2} + \dots + t_{n}\) abbreviates the term \((\dots(t_{1} + t_{2}) + \dots ) + t_{n}\) and \(nt\) denotes the term
\[
  \underbrace{t + t+ \cdots  + t}_{\text{\(n\) times}}.
\]

\begin{defi}[Literal, Clause, Clause set]
An \(L\) formula \(l(\vec{x})\) is called an \(L\) literal if it is an atom or the negation of an atom.
An \(L\) formula \(C( \vec{x} )\) is said to be an \(L\) clause if it of the form \(\forall \vec{y} \, \bigvee_{i = 1}^k l_i\) where \(l_i(\vec{x}, \vec{y})\) with \(i \in \{1,\dots,k\}\) is a literal. 
An \(L\) formula \(S( \vec{x} )\) is called an \(L\) clause set if it is of the form \(\bigwedge_{i = 1}^m C_i\) where \(C_i\) with \(i \in \{1,\dots,m\}\) is an \(L\) clause.
\end{defi}
When the language \(L\) is clear from the context we simply say literal, clause, and clause set instead of \(L\) literal, \(L\) clause, and \(L\) clause set.

Let \(\varphi\) and \(\psi\) be formulas, then by \(\varphi \rightarrow \psi\) we abbreviate the formula \(\neg \varphi \vee \psi\).
In order to save some parentheses we use the standard convention that \(\rightarrow\) associates to the right and has a lower precedence than \(\neg\), \(\vee\), and \(\wedge\).
For example, the formulas \(\neg \varphi \rightarrow \psi\), \(\varphi \wedge \theta \rightarrow \psi\), \(\varphi \rightarrow \theta \vee \psi\) are to be read as \((\neg \varphi) \rightarrow \psi\), \((\varphi \wedge \theta) \rightarrow \psi\) and \(\varphi \rightarrow (\theta \vee \psi)\), respectively.
For the sake of readability we will often not distinguish between a clause and an equivalent formula up to negation normal form.
For example if \(a_{1}, \dots, a_{n}\), and \(b_{1}, \dots, b_{m}\) are atoms, then we also call the formula
\[
  (a_{1} \wedge \dots \wedge a_{n}) \rightarrow (b_{1} \vee \dots \vee b_{m})
\]
a clause, because its negation normal form is the clause \(\bigvee_{i =1}^{n}\neg a_{i} \vee \bigvee_{i = 1}^{m}b_{i}\).
We usually present a concrete clause set \(\bigwedge_{i = 1}^{n}C_{i}\) as the list of clauses \(C_{1}\), \dots, \(C_{n}\).
Similarly, we will usually present the axioms of a theory as a list of formulas.
We are now ready to define the notion of clause set cycles and the related notion of refutability by a clause set cycle.
\begin{defi}
An \(L\) clause set \(S(\Param)\) is called an \(L\) clause set cycle if it satisfies the following conditions
\begin{align}
S( \Succ \Param ) & \models S( \Param ), \label{cond:csc:1} \\
S( \Zero )        & \models \bot. \label{cond:csc:2}
\end{align}
\end{defi}
Note that clause set cycles do not operate in some background theory.
However, a clause set cycle may contain clauses without free variables and these clauses act as a background theory.
A clause set cycle has a natural interpretation as an argument by infinite descent, which we will later explain in terms of induction.
Before that, we introduce the notion of refutation by a clause set cycle.
A refutation by a clause set cycle consists of a clause set cycle and a case distinction.
\begin{defi}
  We say that an \(L\) clause set \(R(\Param)\) is refuted by an \(L\) clause set cycle \(S(\Param)\) if there exists a natural number \(n\) such that
  \begin{align}
    \label{cond:refcsc:1} R(\Succ^n \Param) & \models S(\Param), \\
    \label{cond:refcsc:2} R( \numeral{k} )  & \models \bot, \ \text{for all $k \in \{0, \dots, n - 1\}$}. 
  \end{align}
\end{defi}
If the language is clear from the context we simply speak of clause set cycles and of clause sets refuted by a clause set cycle.
Let us consider an example in order to clarify the notions of clause set cycle and refutation by a clause set cycle.
\begin{exa}
  \label{example:refutation_by_clause_set_cycle}
  Let \(\PredP\) and \(\PredQ\) be unary predicate symbols over the sort \(\sortNat\) and let \(R(\Param)\) be the clause set consisting of the clauses
  \begin{gather*}
    \PredQ( \Param ), \\
    \neg \PredQ(\Zero), \\
    \forall x \, (\PredQ( \Succ x ) \rightarrow \PredP( x )), \\
    \neg \PredP( \Zero ), \\
    \forall x \, (\PredP( \Succ x ) \rightarrow \PredP( x )).
  \end{gather*}
  Our goal is to refute the clause set \(R(\eta)\) by a clause set cycle.

  First we claim that \(S(\Param) \coloneqq \PredP(\Param) \wedge \neg \PredP(\Zero) \wedge \forall x \, (\PredP(\Succ x) \rightarrow \PredP(x))\) is a clause set cycle.
  We start by showing that \(S(\Zero) \models \bot\).
  Observe that \(S(\Zero) = P(\Zero) \wedge \neg P(\Zero) \wedge \forall x \, (\PredP(\Succ x) \rightarrow \PredP(x))\).
  Hence we have \(S(\Zero) \models \bot\).
  It remains to show that \(S(\Succ \eta) \models S(\eta)\).
  Since \(S(\Succ \eta) = P(\Succ \eta) \wedge \neg P(0) \wedge \forall x \, (\PredP(\Succ x) \rightarrow \PredP(x))\), it suffices to show that \(S(\Succ \eta) \models \PredP(\eta)\).
  We have \(S(\Succ \eta) \models P(\Succ \eta)\) and \(S(\Succ \eta) \models \forall x \, (\PredP(\Succ x) \rightarrow \PredP(x))\), hence \(S(\Succ \eta) \models \PredP(\Succ \eta) \rightarrow \PredP(\eta)\) and by modus ponens we obtain \(S(\Succ \eta) \models P(\eta)\).
  Hence \(S(\Param)\) is a clause set cycle.

  Now we claim that \(R(\eta)\) is refuted by the clause set cycle \(S(\eta)\).
  It suffices to show the entailments \(R(\Zero) \models \bot\) and \(R(\Succ \eta) \models S(\eta)\).
  By definition of \(R(\eta)\) we have \(R(\eta) \models \PredQ(\eta)\) and \(R(\eta) \models \neg \PredQ(\Zero)\), hence \(R(\Zero) \models \PredQ(\Zero)\) and \(R(\Zero) \models \neg \PredQ(\Zero)\).
  Therefore \(R(\Zero) \models \bot\).
  For the second entailment we have \(R(\Succ \eta) \models \PredQ(\Succ \eta)\) and \(R(\Succ \eta) \models \forall x \, (\PredQ(\Succ x) \rightarrow \PredP(x))\), thus \(R(\Succ \eta) \models \PredQ(\Succ \eta) \rightarrow \PredP(\eta)\).
  Hence \(R(\Succ \eta) \models \PredP(\eta)\).
  We thus have \(R(\Succ \eta) \models S(\eta)\).
  Hence \(R(\Param)\) is refuted by the clause set cycle \(S(\Param)\) as claimed.
\end{exa}
The cycles introduced in \cite{kersani2013combining} by Kersani and Peltier are parameterized by two natural numbers \(i \geq 0\) and \(j \geq 1\), that control the argument by infinite descent.
The number \(i\)---the offset---is the number at which the argument by infinite descent stops, and \(j\)---the step---is the number by which the descent proceeds.
Our clause set cycles do not have such parameters and therefore appear somewhat restrictive.
Let us now demonstrate how to formulate the concepts of step and offset for clause set cycles.
\begin{defi}
  \label{def:2}
  Let \(S(\Param)\) be an \(L\) clause set and \(i, j \in \mathbb{N}\) with \(j \geq 1\). We say that \(S(\Param)\) is an \(L\) clause set cycle with offset \(i\) and step \(j\) if
  \begin{align*}
    S(\numeral{k + i}) & \models \bot, \ \text{for \(k = 0, \dots, j - 1\) and,} \\
    S(\Succ^{i + j}\Param) & \models S(\Succ^{i}\Param).
  \end{align*}
  An \(L\) clause set \(R(\Param)\) is refuted by the clause set cycle \(S(\eta)\) with offset \(i\) and step \(j\) if there exists a natural number \(n\) such that
  \begin{align*}
    R(\numeral{k}) & \models \bot, \ \text{for \(k = 0, \dots, n -1\)}, \\
    R(\Succ^{n}\Param) & \models S(\Succ^{i}\Param).
  \end{align*}
\end{defi}
We can now show that clause set cycles simulate clause set cycles with offset and step.
\begin{prop}
  \label{pro:3}
  Let \(R(\Param)\) be an \(L\) clause set.
  If \(R(\Param)\) is refuted by an \(L\) clause set cycle \(S(\Param)\) with offset \(i\) and step \(j\), then \(R(\Param)\) is refuted by an \(L\) clause set cycle.
\end{prop}
\begin{proof}
  It is straightforward to see that the offset \(i\) is inessential, by letting \(T(\Param) \coloneqq S(\Succ^{i}\Param)\) because the \(L\) clause set \(T(\Param)\) clearly is a clause set cycle with offset \(0\) and step \(j\) and refutes \(R(\Param)\) with offset \(0\) and step \(j\).
  In order to show that an arbitrary step \(j\) is inessential as well, we let
  \[U(\Param) \coloneqq \bigvee_{l = 0}^{j - 1}T(\Succ^{l}\Param).\]
  We will show that \(U\) is a clause set cycle.
  To show that \(U\) satisfies \eqref{cond:csc:1}, it suffices to observe that by the assumption, we have \(T(\numeral{k}) \models \bot\) for \(k = 0, \dots, j - 1\). Therefore \(U(\Zero) \models \bot\).
  In order to show that \(U\) satisfies \eqref{cond:csc:2}, we need to consider two cases.
  First let \(l \in \{ 0, \dots, j - 2\}\), then we have \(T(\Succ^{l + 1}\Param) \models T(\Succ^{l + 1}\Param)\), which implies \(T(\Succ^{l}\Succ\Param) \models U(\Param)\).
  Now let \(l = j - 1\), then we have \(T(\Succ^{j}\Param) \models T(\Param)\) by the assumption and, thus, \(T(\Succ^{l}\Succ\Param) \models U(\Param)\).
  Therefore, \(U\) is a clause set cycle and \(R\) is refuted by \(U\).
\end{proof}
Clause set cycles thus abstract parameters such as offset and step width and therefore simplify a formal analysis.

Let \(\psi(x, \vec{z})\) be a formula where \(x\) is a variable of sort \(\sortNat\), then the structural induction axiom \(\formulaInd{x}{\psi}\) is defined by
\[
    \formulaInd{x}{\psi} \coloneqq \forall \vec{z} \, (\psi(\Zero, \vec{z}) \rightarrow \forall x \, (\psi(x, \vec{z}) \rightarrow \psi(\Succ x, \vec{z})) \rightarrow \forall x \, \psi(x, \vec{z})).
\]
By an \(\E{1}\) formula we understand a formula of the form \(\exists \vec{x} \, \varphi(\vec{x}, \vec{y})\), where \(\varphi\) is quantifier-free. The notion of \(\A{1}\) formulas is defined dually to \(\E{1}\) formulas. We will now introduce the two theories of induction that are of interest for the study of the formalism of clause set cycles.
\begin{defi}
The theories \(\theoryIndE{L}{1}\) and \(\theoryIndOpen{L}\) are given by
\begin{gather*}
  \theoryIndE{L}{1} \coloneqq \{ \formulaInd{x}{\psi} \mid \text{$\psi(x, \vec{z})$ is an $\E{1}$, $L$ formula with $x : \sortNat$}\}, \\
  \theoryIndOpen{L} \coloneqq \{ \formulaInd{x}{\psi} \mid \text{$\psi(x, \vec{z})$ is a quantifier-free, $L$ formula with $x : \sortNat$}\}.
\end{gather*}
\end{defi}
Whenever the language \(L\) is clear from the context or irrelevant, we will write \(\theoryIndEL{1}\), \(\theoryIndOpenL\) instead of \(\theoryIndE{L}{1}\), \(\theoryIndOpen{L}\).
Let \(\varphi(x)\) be a formula with \(x\) of sort \(\sortNat\), then we say that \(\varphi\) is inductive if \(\vdash \varphi(\Zero)\) and \(\varphi(x) \vdash \varphi(\Succ x)\).

Let us now consider how the notions of clause set cycles and refutability by a clause set cycle relate to provability in theories of induction.
Let \(S(x)\) be a clause set cycle, then by \eqref{cond:csc:1}, \eqref{cond:csc:2}, and by the completeness of first-order logic, we obtain \(\vdash \neg S(\Zero)\) and \(\neg S(x) \vdash \neg S(\Succ x)\).
In other words, the formula \(\neg S(x)\) is inductive.
Since \(S\) is a clause set, \(S\) is logically equivalent to a \(\A{1}\) formula, hence \(\neg S(x)\) is logically equivalent to an \(\E{1}\) formula.
Therefore, we have:
\begin{prop}
\label{lemma:clause-set-cycles-are-unsatisfiable}
Let \(S(x)\) be a clause set cycle, then we have \(\theoryIndEL{1} \vdash \neg S( x )\).
\end{prop}
Intuitively a refutation by a clause set cycle consists of a clause set cycle and a case distinction.
The case distinction with \(n \in \mathbb{N}\) cases can be formalized as follows:
\[ 
  C_n(x) \coloneqq \left(\bigvee_{i = 0}^{n - 1} x = \numeral{i}\right) \vee \exists y \, x = \Succ^n(y).
\]
Since \(C_{n}\) is clearly inductive and logically equivalent to an \(\E{1}\) formula, the formula \(C_{n}\) is provable with \(\E{1}\) induction.
So we have:
\begin{lem}
\label{lemma:induction-implies-case-distinction}
Let \(n \in \mathbb{N}\), then \(\theoryIndEL{1} \vdash C_n\).
\end{lem}
Let now \(R(x)\) be a clause set refuted by a clause set cycle \(S(x)\).
Then there exists a natural number \(n \in \mathbb{N}\) such that \(R\) and \(S\) satisfy the conditions \eqref{cond:refcsc:1} and \eqref{cond:refcsc:2}.
We thus have
\begin{gather}
  \label{upper_bound_ground_case} \vdash \neg R(\numeral{i}), \  i = 0, 1, \dots, n - 1, \\
  \label{upper_bound_variable_case} \neg S(x) \vdash \neg R(\Succ^{n}(x)).
\end{gather}
By the Lemma above we can proceed in \(\theoryIndEL{1}\) by case distinction on the variable \(x\).
If \(x = \numeral{i}\), then we obtain \(\neg R(x)\) by \eqref{upper_bound_ground_case}.
Otherwise if \(x = \Succ^{n}(x')\) for some \(x'\), then by Proposition \ref{lemma:clause-set-cycles-are-unsatisfiable} and \eqref{upper_bound_variable_case} we have \(\neg R(\Succ^{n}(x'))\), thus \(\neg R(x)\).
We therefore obtain:
\begin{thm}
  \label{theorem:upper_bound}
  If a clause set \(R(x)\) is refuted by a clause set cycle, then
  \[\theoryIndEL{1} \vdash \neg R(x).\]
\end{thm}
Refutability by a clause set cycle is thus contained in the theory of \(\E{1}\) induction.
Therefore, methods for \ac{AITP} that are based on clause set cycles can not prove statements that require induction on formulas with two or more quantifier alternations.
This limitation is due to clause set cycles operating on clause sets instead of some larger set of formulas.
Similar limitations may apply to other \ac{AITP} methods that extend saturation-based provers by induction mechanisms that involve only clauses.

\section{Open Induction and Triangular numbers}
\label{sec:org54c1dd2}
In this section we will provide a finite, universal axiomatization of a theory of triangular numbers with open induction.
This finite axiomatization of the theory of triangular numbers will be used in Section \ref{sec:org567ac08}, to show that there exists a clause set that is refutable by a clause set cycle but that is not refutable by open induction.
The result presented in this section is a generalization of the finite axiomatization for multiplication-free arithmetic with open induction provided by Shoenfield in \cite{shoenfield1958}.

Let \(n \in \mathbb{N}\), then by \(\triangle_{n}\) we denote the \(n\)-th triangular number \(\sum_{i = 0}^n i = n(n+1)/2\).
By \(L_{\TriSymbol}\) we denote the one-sorted first-order language consisting of the function symbols \(\Zero/0\), \(\Succ/1\), \(\Pred/1\), \(+/2\), and the binary predicate symbol \(\TriSymbol\).
The predicate symbol \(\TriSymbol\) will be written in infix notation.
\begin{defi}
By \(\TDelta\) we denote the theory axiomatized by
\begin{align*}
 \forall x \, \Succ x & \neq \Zero. \label{A:1} \tag{$\mathrm{A1}$} \\
 \Pred \Zero & = \Zero. \label{A:2} \tag{$\mathrm{A2}$} \\
 \forall x \, \Pred \Succ x & = x. \label{A:3} \tag{$\mathrm{A3}$} \\
 \forall x \, x + \Zero & = x. \label{A:4} \tag{$\mathrm{A4}$} \\
 \forall x \, \forall y \, x + \Succ y & = \Succ ( x + y ). \label{A:5} \tag{$\mathrm{A5}$} \\
 {\Zero} & \ \TriSymbol \ {\Zero}. \label{A:6} \tag{$\mathrm{A6}$} \\
 \forall x \, \forall y \, ({x} \ \TriSymbol \ {y} & \rightarrow {\Succ x} \ \TriSymbol \ {\Succ x + y}). \label{A:7:1} \tag{$\mathrm{A7}$} \\
 \forall x \, \forall y \, ({\Succ x} \ \TriSymbol \ {\Succ x + y} & \rightarrow {x} \ \TriSymbol \ { y}). \label{A:7:2} \tag{$\mathrm{A8}$} \\
 \forall x \, \forall y \, \forall z \, ({x} \ \TriSymbol \ {y} \wedge {x} \ \TriSymbol \ {z} & \rightarrow y = z). \label{A:8} \tag{$\mathrm{A9}$}
\end{align*}
\end{defi}
The standard model for this theory is denoted by \(\mathbb{N}^{\TriSymbol}\). The model \(\mathbb{N}^{\TriSymbol}\) interprets the symbols \(\Zero, \Succ, \Pred, +\) in the natural way. The predicate symbol \(\TriSymbol\) is interpreted as the graph of the triangle function i.e.\ the function associating with each natural number \(n\) the triangular number \(\triangle_{n}\).
\begin{lem}
\label{lemma:B1-to-B4-provable-with-open-induction}
The theory \(\TDelta + \theoryIndOpen{L_{\TriSymbol}}\) proves the following formulas
\begin{align*}
\forall x \, (x \neq \Zero & \rightarrow x = \Succ \Pred x). \label{B:1} \tag{$\mathrm{B1}$} \\
\forall x \, x + y & = y + x. \label{B:2} \tag{$\mathrm{B2}$} \\
\forall x \, \forall y \, \forall z \, (x + y) + z & = x + (y + z). \label{B:3} \tag{$\mathrm{B3}$} \\
\forall x \, \forall y \, \forall z \, (x + y = x + z & \rightarrow y = z). \label{B:4} \tag{$\mathrm{B4}$}
\end{align*}
\end{lem}
The axiom \eqref{A:8} is not redundant for the axiomatization of the theory \(\TDelta + \theoryIndOpen{L_{\TriSymbol}}\).
This can be seen by taking the standard interpretation of the symbols \(\Zero\), \(\Succ\), \(\Pred\), \(+\) and by interpreting \(\TriSymbol\) as the set \(\{(n, \triangle_{n}) \mid n \in \mathbb{N}\} \cup \{ (n, \triangle_{n} + 1) \mid n \in \mathbb{N}\}\).
\begin{defi}
By \(\TDeltaI\) we denote the theory \(\TDelta + \mathrm{B1} + \mathrm{B2} + \mathrm{B3} + \mathrm{B4}\).
\end{defi}
Our axiomatization of the theory \(\TDeltaI\) is not minimal, because \eqref{A:3} can be derived from the other axioms.
Another simple but important observation is that for every formula \(\varphi(y)\) of the language \(L_{\TriSymbol}\) we have
\begin{equation}
  \label{equiv:1}
  \TDeltaI \vdash \varphi(\Succ\Pred x) \leftrightarrow ( x = \Zero \wedge \varphi(\Succ\Zero)) \vee (x \neq \Zero \wedge \varphi(x)).
\end{equation}
We will now show that every formula in \(\TDeltaI\) is equivalent to some formula that is ``simple'' in the following sense.
\begin{defi}
  We call a term or a formula simple if it does not contain the symbol \(\Pred\).
\end{defi}
In the theory \(\TDeltaI\) simple terms have a convenient equivalent representation.
\begin{lem}
  \label{lem:5}
  Let \(t(x_{1}, \dots, x_{n})\) be a simple term, then there exist \(m_{1}, \dots, m_{n}, k \in \mathbb{N}\) such that \(\TDeltaI \vdash t = m_{1}x_{1} + \dots + m_{n}x_{n} + \numeral{k}\).
\end{lem}
\begin{proof}
  The term \(t\) can be rewritten with \(\eqref{A:4}\), \(\eqref{A:5}\), \(\eqref{B:2}\), and \(\eqref{B:3}\) into the form \(m_{1}x_{1} + \dots + m_{n}x_{n} + \numeral{k}\).
\end{proof}
Whenever we are working in the context of the theory \(\TDeltaI\) we will often and implicitly assume that a simple term is of the form \(m_{1}x_{1} + \dots + m_{n}x_{n} + \numeral{k}\).
\begin{prop}
  \label{proposition:formula_equivalent_to_simple_formula}
Let \(\varphi\) be a formula, then there exists a simple formula \(\psi\) such that \(\TDeltaI \vdash \varphi \leftrightarrow \psi\).
\end{prop}
\begin{proof}
  We start with an important observation that we shall repeatedly use throughout the proof.
  Let \(t\) be a term containing the symbol \(\Pred\).
  Now we work in the theory \(\TDeltaI\).
  Then by using \eqref{A:5}, \eqref{B:3}, and \eqref{B:4}, it is possible to permute the symbol \(\Succ\) in the term \(\Succ(t)\) inwards until it is directly above an occurrence of the symbol \(\Pred\).

  The proof consists of three main steps.
  First, we eliminate occurrences of \(\Pred\) in the left-hand arguments of triangle atoms.
  After that, we eliminate in the resulting formula all occurrences of \(\Pred\) in right-hand arguments of triangle atoms without introducing \(\Pred\) to left-hand arguments of triangle atoms.
  Finally, we eliminate the symbol \(\Pred\) from equational atoms without introducing \(\Pred\) to the triangle atoms.

  Let \(\psi\) be a formula.
  We will now show that there exists a formula \(\psi'\) such that \(\TDeltaI \vdash \psi \leftrightarrow \psi'\) and moreover \(\psi'\) does not contain an occurrence of the symbol \(\Pred\) in the left-hand arguments of triangle atoms.
  We start by defining a measure on formulas that will make the argument more apparent.
  Let us define \(\#_{1}^{n}(\psi)\) to be the number of triangle atoms in \(\psi\) having exactly \(n\) occurrences of \(\Pred\) in their left-hand argument.
  We furthermore let
  \[
    N_{1}^{\psi} \coloneqq \max \{ n \in \mathbb{N} \mid \#_{1}^{n}(\psi) \neq 0\}.
  \]
  Now we can define the measure \(\#_{1}(\psi) \coloneqq (N_{1}^{\psi}, \#_{1}^{N_{1}^{\psi}}(\psi))\).
  We proceed by induction on the measure \(\#_{1}(\psi)\) with respect to the natural order on pairs of natural numbers.
  If \(\#_{1}(\psi) = (0, m)\) for some \(m \in \mathbb{N}\), then \(\psi\) does not contain the symbol \(\Pred\) in the left-hand side argument of its triangle atoms.
  Therefore \(\psi\) is already the desired formula.
  Otherwise if \(N_{1}^{\psi} \neq 0\), then \(\#_{1}(\psi) = (N_{1}^{\psi}, \#_{1}^{N_{1}^{\psi}}(\psi))\) with \(\#_{1}^{N_{1}^{\psi}}(\psi) \neq 0\).
  Hence there exists a triangle atom \(\varphi \coloneqq \Tri{t_{1}}{t_{2}}\) of \(\psi\) having exactly \(N_{1}^{\psi}\) occurrences of \(\Pred\) in its left-hand argument.
  Let \(\varphi' \coloneqq \Tri{\Succ(t_{1})}{\Succ(t_{1}) + t_{2}}\), then by \eqref{A:7:1} and \eqref{A:7:2} we have \(\TDeltaI \vdash \varphi \leftrightarrow \varphi'\).
  Obtain a \(\TDeltaI\) equivalent atom \(\varphi^{\prime\prime}\) from \(\varphi'\) by permuting in the left-hand argument \(\Succ(t_{1})\) the symbol \(\Succ\) inwards, as described above, until it is directly above an occurrence of \(\Pred\).
  Now apply \eqref{equiv:1} to \(\varphi^{\prime\prime}\) in order to obtain a \(\TDeltaI\) equivalent formula \(\varphi^{\prime\prime\prime}\).
  Note that \(\varphi^{\prime\prime\prime}\) has four atoms---two equational atoms and two triangle atoms.
  The triangle atoms both contain exactly \(N_{1}^{\psi}-1\) occurrences of \(\Pred\) in their left-hand arguments.
  Let \(\psi'\) be obtained from \(\psi\) by replacing \(\varphi\) by \(\varphi^{\prime\prime\prime}\).
  We then have \(\TDeltaI \vdash \psi' \leftrightarrow \psi\).
  If \(N_{1}^{\psi'} = N_{1}^{\psi}\), then by the above we have \(\#_{1}^{N_{1}^{\psi'}}(\psi') = \#_{1}^{N_{1}^{\psi}}(\psi') = \#_{1}^{N_{1}^{\psi}}(\psi) - 1\). Hence, \(\#_{1}(\psi') < \#_{1}(\psi)\).
  Otherwise we have \(N_{1}^{\psi'} < N_{1}^{\psi}\) and therefore  \(\#_{1}(\psi') < \#_{1}(\psi)\).
  In any case we have \(\#_{1}(\psi') < \#_{1}(\psi)\), thus we obtain the desired formula by applying the induction hypothesis to \(\psi'\).

  By the previous step we can now assume that we work with formulas that do not contain \(\Pred\) in the left-hand arguments of triangle atoms.
  We will now eliminate all the occurrences of \(\Pred\) in right hand arguments of triangle atoms.
  To accomplish this we proceed as in the first step by defining \(\#_{2}^{n}(\psi)\), \(N_{2}^{\psi}\), and \(\#_{2}(\psi)\) in analogy to \(\#_{1}^{n}(\psi)\), \(N_{1}^{\psi}\), and \(\#_{1}(\psi)\) for a formula \(\psi\).
  For a formula \(\psi\) without \(\Pred\) in the left-hand arguments of triangle atoms we proceed by induction on \(\#_{2}(\psi)\) as in the first step.
  If \(N_{2}^{\psi} = 0\), then \(\psi\) does not contain \(\Pred\) its triangle atoms.
  Thus \(\psi\) is the desired formula.
  Otherwise if \(N_{2}^{\psi} \neq 0\), then \(\#_{2}^{N_{2}^{\psi}}(\psi) \neq 0\), that is, there exists an atom \(\varphi \coloneqq \Tri{t_{1}}{t_{2}}\) of \(\psi\) with exactly \(N_{2}^{\psi}\) occurrences of \(\Pred\) in \(t_{2}\).
  We let \(\varphi' \coloneqq \Tri{\Succ(t_{1})}{\Succ(t_{1} + t_{2})}\).
  Then \(\TDeltaI \vdash \varphi \leftrightarrow \varphi'\).
  Since \(t_{1}\) is assumed to be free of \(\Pred\), the term \(\Succ(t_{1} + t_{2})\) contains \(N_{2}^{\psi}\) occurrences of \(\Pred\).
  Now we proceed in analogy to the first step by moving the outermost \(\Succ\) inwards, thus allowing us to eliminate one occurrence of \(\Pred\) by making use of the equivalence \(\eqref{equiv:1}\).
  We therefore obtain a formula \(\psi'\) with \(\TDeltaI \vdash \psi \leftrightarrow \psi'\) and \(\#_{2}(\psi') < \#_{2}(\psi)\).
  
  We can now assume that formulas we work with do not contain \(\Pred\) in their triangle atoms.
  It remains to show that we can eliminate the occurrences of \(\Pred\) from equational atoms without introducing triangle atoms containing \(\Pred\).
  Let \(\psi\) be a formula, then we define \(\#_{3}^{n}(\psi)\) to be the number of equational atoms of \(\psi\) containing exactly \(n\) occurrences of \(\Pred\).
  We moreover define \(N_{3}^{\psi}\) and \(\#_{3}(\psi)\) in analogy to \(N_{1}^{\psi}\) and \(\#_{1}(\psi)\).
  The induction base is again trivial.
  For the induction step, we have \(N_{3}^{\psi} \neq 0\) and \(\#_{3}^{N_{3}^{\psi}}(\psi) \neq 0\).
  Therefore there exists an equational atom \(\varphi \coloneqq (t_{1} = t_{2})\) in \(\psi\) with exactly \(N_{3}^{\psi}\) occurrences of \(\Pred\).
  We proceed similarly to the previous two steps above, by first replacing an the atom \(\varphi\) by the \(\TDeltaI\) equivalent atom \(\Succ(t_{1}) = \Succ(t_{2})\), then moving \(s\) inwards, and finally applying the equivalence \eqref{equiv:1}.
  Let \(\psi'\) be the resulting formula.
  It is easy to see that this does not introduce triangle atoms containing \(\Pred\) and moreover we have \(\#_{3}(\psi') < \#_{3}(\psi)\).
  We can thus apply the induction hypothesis in order to obtain the desired formula.
\end{proof}
At this point the reader might wonder why the authors chose to include a function symbol for the predecessor function and go through the technicalities of eliminating the symbol \(\Pred\) from formulas instead of providing additional axioms.
The idea is to work with purely universal axiomatizations so that one can in particular apply Herbrand's theorem.
In order to avoid the symbol \(\Pred\) we would have to include an axiom containing an existential quantifier such as \(\forall x \, \exists y \, (( x = \Zero \wedge y = \Zero) \vee (x \neq \Zero \wedge x = \Succ(y)))\).
As a consequence, we would later run into similar technicalities when dealing with the this existential quantifier.
In fact, the function symbol \(\Pred\) is just a Skolem function introduced for the axiom given above. 

We have now everything at hand to start with the model theoretic considerations of the theory \(\TDeltaI\).
In the following we fix an arbitrary model \(\mathcal{M}\) of the theory \(\TDeltaI\).
Our aim is to show that \(\mathcal{M}\) is also a model of open induction over the language \(L_{\TriSymbol}\).
By \(\ZeroInter\), \(\SuccInter\), \(\PredInter\), \(\PlusInter\), and \(\TriInterSymbol\) we denote the respective interpretations of the symbols \(\Zero\), \(\Succ\), \(\Pred\), \(+\), and \(\TriSymbol\) in the model \(\mathcal{M}\).
We start with a few simple observations about the structure of the model \(\mathcal{M}\).
\begin{defi}
  \label{definition:comparability}
  \label{definition:less_than}
  Let \(a,b \in \mathcal{M}\), then we define \(a \prec b\) if there exists \(n \geq 1\) such that \(\SuccInter^{n}a = b\).
  Accordingly we define \(a \preceq b\) if \(a \prec b\) or \(a = b\).
  We say that \(a\) and \(b\) are comparable (in symbols \(a \sim b\)), if \(a \preceq b\) or \(b \preceq a\).
\end{defi}
It is not hard to see that the relation \(\preceq\) is a partial order and that \(\sim\) is an equivalence relation.
Let \(a\) be an element of \(\mathcal{M}\), then by \([a]\) we denote the equivalence class of \(a\) under \(\sim\).
It is easy to see that \(\preceq\) is total on \([a]\).
Hence classes of comparable elements together with \(\preceq\) form chains.
Let us now look a bit more closely at these chains.
As a simple consequence of \eqref{B:4} we have \(\forall x \, x + \numeral{k} \neq x\) for all \(k \geq 1\).
Therefore the chains of comparable elements are infinite.
Consider now the class of elements comparable with \(\ZeroInter\) and let \(a\) be an element comparable with \(\ZeroInter\).
Then by \eqref{A:1} we have \(a = \SuccInter^{m}\ZeroInter\) for some natural number \(m\), that is, \(\ZeroInter\) is the least element of this chain. The chain of elements comparable with \(\ZeroInter\) thus looks as follows:
\[
  \ZeroInter \prec \SuccInter^{1}\ZeroInter \prec \SuccInter^{2} \ZeroInter \prec \SuccInter^{3}\ZeroInter \prec \dots.
\]
This chain is isomorphic to the chain of natural numbers with \(\leq\) and the addition.
This is why we will call the elements comparable with \(\ZeroInter\), the standard elements (of \(\mathcal{M}\)).
Elements that do not belong to this chain are called non-standard elements (of \(\mathcal{M}\)).

Before we have a look at the structure of a chains of non-standard elements let us summarize some basic properties of \(\mathcal{M}\).
\begin{lem}
  Let \(a, b\) be elements of \(\mathcal{M}\)
  \label{lem:1}
  \begin{enumerate}
  \item If \(a\) is a non-standard element, then \(\PredInter a\) is a non-standard element.
  \item If \(a\) is a non-standard element and \(m \in \mathbb{N}\), then \(a = \SuccInter^{m}\PredInter^{m}a\).
  \item The element \(a \PlusInter b\) is a standard element if and only if \(a\) and \(b\) are standard elements.
  \end{enumerate}
\end{lem}
\begin{proof}
  For (1) let \(a\) be a non-standard element of \(\mathcal{M}\), then \(a \neq \ZeroInter\). Thus by \eqref{B:1} the element \(\PredInter a\) is also a non-standard element.

  For (2) if \(a\) is non-standard, then \(a \neq 0\) and by \eqref{B:1} we have \(a = \SuccInter\PredInter a\). By (1) the element \(\PredInter a\) is non-standard, so we have \(\PredInter = \SuccInter\PredInter(\PredInter a)\) thus \(a = \SuccInter^{2}\PredInter^{2}a\) and so on.
  
  For (3) consider now an element of the form \(a \PlusInter b\).
If \(a\) and \(b\) are both standard elements, then it is clear that \(a \PlusInter b\) is a standard element.
Now suppose that \(a \PlusInter b\) is a standard element and suppose without loss of generality that \(a\) is not a standard element.
Then there exists \(m \in \mathbb{N}\) such that \(\SuccInter^{m}\ZeroInter = a \PlusInter b = b \PlusInter a = b \PlusInter \SuccInter^{m + 1}\PredInter^{m+1}a\).
By \eqref{A:3} we obtain \(\ZeroInter = \SuccInter\PredInter^{m+1}a\) which contradicts \eqref{A:1}.
Hence \(a\) must also be a standard element.  
\end{proof}
Let us now consider the chain of elements comparable with a non-standard element \(a\) of \(\mathcal{M}\).
Let \(n, m\) be natural numbers with \(n < k\), then by Lemma \ref{lem:1} we have \(\SuccInter^{n}\PredInter^{k}a = \PredInter^{k -n}a\).
Hence we have \(\PredInter^{k}a \prec \PredInter^{k -n}a\).
Let \(b\) be comparable with \(a\), then either \(a = \SuccInter^{m}b\) or \(b = \SuccInter^{m}a\) i.e.\ \(b = \PredInter^{m}a\) or \(b = \SuccInter^{m}a\) for some natural number \(m\).
The chain thus has the following structure:
\[
  \dots \prec \PredInter^{3}a \prec \PredInter^{2}a \prec \PredInter^{1}a \prec a \prec \SuccInter^{1}a \prec \SuccInter^{2}a \prec \SuccInter^{3}a \prec \dots.
\]
The chain of elements comparable with \(a\) is isomorphic to the integers with the order \(<\).

We define the language \(L_{\TriSymbol}(\mathcal{M})\) to be the extension of the language \(L_{\TriSymbol}\) by a constant symbol \(a\) for every element \(a\) of \(\mathcal{M}\).
We let the \(L_{\TriSymbol}\) structure \(\mathcal{M}\) interpret \(L_{\TriSymbol}(\mathcal{M})\) formulas by interpreting for every element \(a\) of \(\mathcal{M}\) the constant \(a\) as itself.
The language \(L_{\TriSymbol}(\mathcal{M})\) will be especially convenient when we need to insert elements of \(\mathcal{M}\) into \(L_{\TriSymbol}\) terms and \(L_{\TriSymbol}\) formulas.
Let \(\varphi(x)\) be an \(L_{\TriSymbol}(\mathcal{M})\) formula.
We call an element \(a\) of \(\mathcal{M}\) a solution of \(\varphi\) if \(\varphi(a)\) is true in \(\mathcal{M}\).
Similarly we call \(\varphi\) valid in \(\mathcal{M}\) if \(\varphi(a)\) is true in \(\mathcal{M}\) for all elements \(a\) of \(\mathcal{M}\).
In the following we will show the crucial observation that simple atomic formulas are either valid in \(\mathcal{M}\) or have only finitely many pairwise comparable solutions.
\begin{prop}
  \label{proposition:simple_atoms_are_valid_or_have_finitely_many_comparable_solutions}
  Let \(\varphi(x,\vec{y})\) be a simple atomic formula, \(\vec{b}\) a vector of elements of \(\mathcal{M}\), then either \(\varphi(x,\vec{b})\) is valid in \(\mathcal{M}\) or \(\varphi(x, \vec{b})\) has only finitely many pairwise comparable solutions.
\end{prop}
\begin{proof}
  Depending on the form of \(\varphi\) we need to distinguish between two cases.
    If \(\varphi\) is of the form \(s = t\), then clearly \(\varphi(x, \vec{b})\) is equivalent in \(\mathcal{M}\) to \(mx + c = nx + d\) for some \(c, d \in \mathcal{M}\).
  The claim then follows from Lemma 1 in \cite{shoenfield1958}.
  If \(\varphi\) is of the form \(\Tri{s}{t}\), then \(\varphi(x, \vec{b})\) is equivalent in \(\mathcal{M}\) to \(\TriInter{mx + c}{nx + d}\) for some \(n, m \in \mathbb{N}\) and \(c,d \in \mathcal{M}\).
  We need to consider two cases:
  \begin{itemize}
  \item For \(m = 0\), we assume that there are at least two comparable solutions \(e\) and \(\SuccInter^{p}e\) of \(\varphi(x, \vec{b})\) with \(p > 0\).
    We have \(\TriInter{c}{ne + d}\) and \(\TriInter{c}{n \mathbf{S}^pe + d}\).
    Therefore by \eqref{A:8} we have \(n e + d = n \mathbf{S}^{p}e + d\).
    By \eqref{B:2} and \eqref{A:5} we obtain \(n e + d = \mathbf{S}^{np}\mathbf{0} \PlusInter n e + d\).
    By \eqref{B:4} we then have \(\mathbf{0} = \mathbf{S}^{np}\mathbf{0}\).
    Hence we clearly have \(n = 0\).
    Thus \(n e + d = d\), and therefore \(\TriInter{c}{d}\) is true in \(\mathcal{M}\).
    Because of that \(\TriInter{c}{nx + d}\) is valid in \(\mathcal{M}\).
  \item For \(m > 0\), we will show that there are at most two comparable solutions of \(\varphi(x, \vec{b})\).
    We proceed indirectly and assume that there are at least three pairwise comparable solutions \(e\), \(\SuccInter^{p_{1}}e\), and \(\SuccInter^{p_{2}}e\) of \(\varphi(x, \vec{b})\) with \(0 < p_{1} < p_{2}\).
  Since \(e\) is a solution we have \(\TriInter{me \PlusInter c}{ne \PlusInter d}\).
  Let \(i \in \{ 1, 2\}\), then iterating \eqref{A:7:1} and straightforward rewriting we have
  \[
    \TriInter{\SuccInter^{p_{i}m}(me + c)}{\SuccInter^{\triangle_{p_{i}m}}\ZeroInter + p_{i}m(me + c) + ne + d}.
  \]
  Since \(\SuccInter^{p_{i}}e\) is a solution of \(\varphi(x, \vec{b})\) we have \(\TriInter{m(\SuccInter^{p_{i}}e) + c}{n(\SuccInter^{p_{i}}e) + d}\).
  Therefore by \eqref{A:5}, \eqref{A:8}, \eqref{B:2} and \eqref{B:4} we obtain
  \[
    \SuccInter^{\triangle_{p_{i}m}}\ZeroInter + p_{i}m(me + c) = \SuccInter^{np_{i}}\ZeroInter.
  \]
  Thus the element \(me + c\) is a standard element of \(\mathcal{M}\).
  Therefore there exists \(k \in \mathbb{N}\) such that \(me + c = \SuccInter^{k}\ZeroInter\).
  We thus have \(\SuccInter^{\triangle_{p_{i}m}}\ZeroInter + p_{i}m\SuccInter^{k}\ZeroInter = \SuccInter^{np_{i}}\ZeroInter\).
  Hence by \eqref{A:1} and because \(\SuccInter\) is injective we obtain
  \[
    \triangle_{p_{i}m} + p_{i}mk = np_{i}.
  \]
  Hence \(m^{2}p_{1} + m + 2mk = 2n = m^{2}p_{2} + m + 2mk\). But since \(m \neq 0\), this contradicts the assumption that \(p_{1} < p_{2}\). \qedhere
  \end{itemize}
\end{proof}
We are now ready to show that \(\mathcal{M}\) is a model of open induction over the language \(L_{\TriSymbol}\).
The proof is analogous to the proof given in \cite{shoenfield1958}.
For the sake of completeness we outline the main steps of the proof.
\begin{thm}
  \label{theorem:extended_shoenfield_theorem}
  Let \(\mathcal{M}\) be a model of \(\TDeltaI\), then \(\mathcal{M}\) is a model of \(\theoryIndOpen{L_{\TriSymbol}}\).
\end{thm}
\begin{proof}
  Let \(\theta(x, \vec{z})\) be a quantifier-free \(L_{\TriSymbol}\) formula.
  We have to show that \(\mathcal{M} \models \formulaInd{x}{\theta(x, \vec{z})}\).
  By Proposition \ref{proposition:formula_equivalent_to_simple_formula} we can assume without loss of generality that \(\theta\) is a simple formula.
  Let \(\vec{b}\) be a vector of elements of \(\mathcal{M}\) and let \(\theta'(x) \coloneqq \theta(x, \vec{b})\). Assume that \(\theta'(\mathbf{0})\) is true in \(\mathcal{M}\) and \(\theta'(x) \rightarrow \theta'(\mathbf{S}(x))\) is valid in \(\mathcal{M}\).
  Let \(a\) be an arbitrary element of \(\mathcal{M}\).
  If \(a\) is a standard element, then \(a = \SuccInter^{m}\ZeroInter\) for some natural number \(m\). So by applying the induction step repeatedly, starting with \(\theta'(\ZeroInter)\), we obtain \(\theta'(a)\).
  
  Now let us consider the case when \(a\) is non-standard.
  Consider the atoms from which \(\theta'(x)\) is built.
  By Proposition \ref{proposition:simple_atoms_are_valid_or_have_finitely_many_comparable_solutions} there are two types of atoms: those that are valid in \(\mathcal{M}\), and those that have at most finitely many (two) pairwise comparable solutions.
  Valid atoms are true in \(\mathcal{M}\) regardless of the choice of \(x\), so we need to consider only the remaining atoms.
  By successively letting \(x\) be \(a\), \(\PredInter^{1}a\), \(\PredInter^{2}a\), and so on, we will eventually exhaust all the solutions (\(\preceq a\)) of any invalid atom of \(\theta'(x)\).
  In other words, by choosing \(n \in \mathbb{N}\) large enough, the element \(\PredInter^{n}a\) falsifies all the invalid atoms.
  The same technique works for standard elements, starting at \(\ZeroInter\) and successively considering \(\SuccInter^{1}\ZeroInter\), \(\SuccInter^{2}\ZeroInter\), and so forth.
  Now by taking \(m \in \mathbb{N}\) large enough so that both \(\PredInter^{m}a\) and \(\SuccInter^{m}\ZeroInter\) falsify all the invalid atoms, we observe that \(\theta'(\SuccInter^{m}\ZeroInter)\) and \(\theta'(\PredInter^{m}a)\) have the same truth value in \(\mathcal{M}\). Therefore, since \(\theta'(\SuccInter^{m}\ZeroInter)\) is true in \(\mathcal{M}\), also \(\theta'(\PredInter^{m}a)\) is true.
  Since \(\SuccInter^{m}\PredInter^{m}a = a\), we can simply apply the induction step \(m\) times to find that \(\theta'(a)\) is true in \(\mathcal{M}\).
\end{proof}
The finite, universal axiomatizability of \(\TDelta + \theoryIndOpen{L_{\TriSymbol}}\) now follows immediately from Theorem \ref{theorem:extended_shoenfield_theorem}, completeness of first-order logic, and from Lemma \ref{lemma:B1-to-B4-provable-with-open-induction}.
\begin{thm}
\label{theorem:modified-shoenfield-theorem}
Let \(\varphi\) be a formula, then \(\TDelta + \theoryIndOpen{L_{\TriSymbol}} \vdash \varphi\) if and only if \(\TDeltaI \vdash \varphi\).
\end{thm}

\section{Clause Set Cycles and Open Induction}
\label{sec:org567ac08}

In Section \ref{sec:org353db4f} we have shown that refutability by a clause set cycle is contained in the theory of \(\E{1}\) induction.
The next obvious question to ask is whether refutability by a clause set cycle is also contained in the theory of open induction.
In this section we will provide a negative answer to that question by making use of the finite axiomatizability of the theory of triangular numbers with open induction shown in Section \ref{sec:org54c1dd2}.
In order to provide such a negative answer it suffices to provide a clause set which is refutable by a clause set cycle, but that is not refutable by open induction.
A candidate clause set is readily found.
\begin{defi}
  We denote by \(S_{\triangleright}(\Param)\) the clause set consisting of the clauses \eqref{A:4} -- \eqref{A:7:1} and the clause \(\forall y \, \neg \Tri{\Param}{y}\).
\end{defi}
Let us denote by \(L_{\triangleright}'\) the language of the clause set \(S_{\triangleright}\).
The clause set \(S_{\triangleright}\) expresses that the triangle function is not total.
\begin{lem}
  \label{lemma:S_triangle_refutable_with_clause_set_cycles}
The clause set \(S_{\triangleright}(\Param)\) is refutable by a clause set cycle.
\end{lem}
\begin{proof}
  By the soundness of first-order logic it suffices to show that
  \begin{gather}
    S_{\triangleright}(\Zero) \vdash \bot, \label{eq:1} \\
    \intertext{and}
    S_{\triangleright}(\Succ(\Param)) \vdash S_{\triangleright}(\Param). \label{eq:2}
  \end{gather}
  For \eqref{eq:1} we have \(S_{\triangleright}(\Zero) \vdash \Tri{\Zero}{\Zero}\) and \(S_{\triangleright}(\Zero) \vdash \forall y \, \neg \Tri{\Zero}{y}\).
  Hence \(S_{\triangleright}(\Zero) \vdash \bot\).
  For \eqref{eq:2} we assume \(S_{\triangleright}(\Succ(\Param))\).
  The clauses of \(S_{\triangleright}\) not having free variables occur in \(S_{\triangleright}(\Succ(\eta))\), hence we only need to show that \(S_{\triangleright}(\Succ(\eta)) \vdash \forall y \, \neg \Tri{\eta}{y}\).
  Let \(y\) be arbitrary, then obtain \(S_{\triangleright}(\Succ(\eta)) \vdash \neg \Tri{\Succ(\Param)}{\Succ(\Param) + y}\).
  By the contrapositive of \eqref{A:7:1} we have \(\neg \Tri{\Param}{y}\).
  Therefore the clause set \(S_{\triangleright}\) is a clause set cycle.
  Since a clause set cycle is trivially refuted by itself, we are done.
\end{proof}
It now remains to show that \(S_{\triangleright}\) cannot be refuted by open induction.
In order to be able to make use of Theorem \ref{theorem:modified-shoenfield-theorem}, we will now reformulate the clause set \(S_{\triangleright}\) in terms of a theory of triangular numbers.
In the following we denote by \(\TDeltaSimple\) the theory axiomatized by the formulas \eqref{A:4} -- \eqref{A:7:1}.
\begin{lem}
  \label{lemma:clause_set_cycle_reduces_to_theory}
  \(\theoryIndOpen{L_{\triangleright}'} \vdash \neg S_{\TriSymbol}(\Param)\) if and only if \(\TDeltaSimple + \theoryIndOpen{L_{\triangleright}'} \vdash \forall x \, \exists y \, \Tri{x}{y}\).
\end{lem}
\begin{proof}
  We have the following chain of equivalences.
  \begin{align*}
    \label{eq:2}
    & \theoryIndOpen{L_{\triangleright}'} \vdash \neg S_{\TriSymbol}(\Param) \\
    & \Leftrightarrow \theoryIndOpen{L_{\triangleright}'} \vdash \neg (\eqref{A:4} \wedge \dots \wedge \eqref{A:7:1} \wedge \forall y \, \neg \Tri{\eta}{y}) \\
    & \Leftrightarrow \theoryIndOpen{L_{\triangleright}'} \vdash \neg (\eqref{A:4} \wedge \dots \wedge \eqref{A:7:1} \wedge \neg \exists y \, \Tri{\eta}{y}) \\
    & \Leftrightarrow \theoryIndOpen{L_{\triangleright}'} \vdash \eqref{A:4} \wedge \dots \wedge \eqref{A:7:1} \rightarrow \exists y \, \Tri{\eta}{y} \\
    & \Leftrightarrow \theoryIndOpen{L_{\triangleright}'} \vdash \eqref{A:4} \wedge \dots \wedge \eqref{A:7:1} \rightarrow \forall x \, \exists y \, \Tri{x}{y}
  \end{align*}
  By the deduction theorem we thus have \(\theoryIndOpen{L_{\triangleright}'} \vdash \neg S_{\TriSymbol}(\Param)\) if and only if \(\TDeltaSimple + \theoryIndOpen{L_{\triangleright}'} \vdash \forall x \, \exists y \, \Tri{x}{y}\).
\end{proof}
In order to complete the negative answer it clearly suffices to show that \(\TDelta + \theoryIndOpen{L_{\triangleright}} \not \vdash \forall x \, \exists y \, \Tri{x}{y}\).
\begin{prop}
\label{proposition:triangular-numbers-not-total-with-quantifier-free-induction}
\(\TDelta + \theoryIndOpen{L_{\triangleright}} \not \vdash \forall x \, \exists y \, \Tri{x}{y}\).
\end{prop}
\begin{proof}
  We proceed indirectly and assume that \(\TDelta + \theoryIndOpen{L_{\triangleright}} \vdash \forall x \, \exists y \, \Tri{x}{y}\).
  By Theorem \ref{theorem:modified-shoenfield-theorem} we then also have \(\TDeltaI \vdash \forall x \, \exists y \, \Tri{x}{y}\).
  Since \(\TDeltaI\) is a universal theory we can apply Herbrand's theorem to obtain terms \(t_{1}(x), \dots, t_{k}(x)\) such that
  \[
    \TDeltaI \vdash \bigvee_{i = 1}^{k}\Tri{x}{t_{i}(x)}.
  \]
  Clearly \(\mathbb{N}^{\triangleright}\) is a model of \(\TDeltaI\) and the triangle function \(\TriSymbol^{\mathbb{N}^{\triangleright}}\) is quadratic.
  Since the terms \(t_{i}(x)\), with \(i = 1, \dots, k\) describe linear functions in \(\mathbb{N}^{\triangleright}\), there exists \(m \in \mathbb{N}\) such that
  \[
    \mathbb{N}^{\triangleright} \not \models \bigvee_{i = 1}^{k}\Tri{\numeral{m}}{t_{i}(\numeral{m})}.
  \]
  Contradiction!
\end{proof}
We would like to point out that it seems possible to obtain the result of Proposition \ref{proposition:triangular-numbers-not-total-with-quantifier-free-induction} by an alternative argument that relies on an interpretation of the theory \(\TDelta + \theoryIndOpen{L_{\triangleright}}\) in the subtheory \(I\Delta_{0}\) of Peano arithmetic.
The idea is to interpret individuals of \(\TDelta + \theoryIndOpen{L_{\triangleright}}\) as ``unary'' numbers, that is, finite sequences of ``1''s.
In particular \(\Zero\) would be interpreted as the empty sequence, \(\Succ\) as the function that appends a ``1'' to a finite sequence, \(+\) as the function that concatenates two finite sequences, and \(\TriSymbol\) as a \(\Delta_{0}\) formula that defines the graph of the ``unary'' triangle function.
Since \(I\Delta_{0}\) codes finite sequences the translation outlined above is indeed an interpretation.
By Parikh's theorem the provably total functions of \(I\Delta_{0}\) on ``unary'' elements have a linear growth rate, whereas the ``unary'' triangle function has a quadratic growth rate.
Hence, the translation of \(\forall x \, \exists y \, \Tri{x}{y}\) is not provable in \(I\Delta_{0}\) and therefore the formula \(\forall x \, \exists y \, \Tri{x}{y}\) is not provable in \(\TDelta + \theoryIndOpen{L_{\triangleright}}\).
All the notions necessary to develop the argument outlined above can be found in standard textbooks on first-order arithmetic such as \cite{hajek2017}.
\begin{cor}
  \label{cor:1}
  The clause set \(S_{\TriSymbol}(\Param)\) is refutable by a clause set cycle and \(\theoryIndOpen{L_{\triangleright}'} \not \vdash \neg S_{\TriSymbol}(\Param)\).
\end{cor}
To summarize we thus have shown the following theorem.
\begin{thm}
  \label{theorem:clause_set_cycles_not_contained_in_open_induction}
  There exists a language \(L\), and an \(L\) clause set \(S(x)\) refutable by clause set cycles such that \(\theoryIndOpen{L} \not \vdash \neg S(x)\).
\end{thm}
Since refutability by a clause set cycle is not contained in open induction, the next obvious question to ask is whether every clause set that is refutable with open induction is also refutable by a clause set cycle.
We believe that this is not the case.
Intuitively, this can be explained by the following two points: first clause set cycles do not allow for any free variables and secondly clause set cycles only allow for existential quantification.
These two shortcomings of clause set cycles can be demonstrated by the following example.
We assume the usual right recursive definition of the addition from which we want to prove the sentence \(\varphi \equiv \forall x \,  x + (x + x) = (x + x) + x\).
To prove \(\varphi\) with open induction, we first prove by open induction on the variable \(y\) the inductive formula \(\psi(x) \equiv \forall y \, x + (x + y) = (x + x) + y\).
The ``lemma'' \(\psi\) can now be used to prove \(\varphi\) by instantiating the universally quantified variable \(y\) by \(x\).
In this example both ideas mentioned above came into play, that is, the ``lemma'' \(\psi\) contains a free variable and it contains a universal quantifier that is actually used to prove \(\varphi\).
\begin{conj}
  \label{conjecture:clause_set_cycles_incomparable_with_open_induction}
  There exists a language \(L\), and an \(L\) clause set \(S(x)\) such that \(\theoryIndOpen{L} \vdash \neg S(x)\) but \(S(x)\) is not refutable by a clause set cycle.
\end{conj}

\section{The $n$-Clause Calculus: A Case Study}
\label{sec:orga779b6c}

In this section we will use the notion of clause set cycle in order to derive results about a concrete approach for \ac{AITP}---the n-clause calculus.
The n-clause calculus is a formalism for \ac{AITP} that was introduced by Kersani and Peltier in \cite{kersani2013combining}.
This calculus enhances the superposition calculus \cite{bachmair1994}, \cite{nieuwenhuis2001}---a refinement of resolution-paramodulation calculi---by a cycle detection mechanism.
This mechanism detects a certain type of cyclic dependencies between the clauses that are derived during the saturation process.
Such a cyclic dependency represents an argument by infinite descent and, therefore, represents an unsatisfiable subset of the derived clauses.
Once such a cycle is detected the refutation is terminated.
The n-clause calculus operates on a syntactically restricted fragment of the logical formalism presented in Section \ref{sec:org353db4f}.
The languages in this section are assumed to contain at least one other sort, say \(\iota\), besides the sort \(\sortNat\) of natural numbers.
Furthermore, the languages should not contain any other function symbols of range \(\sortNat\) besides \(\Zero\) and \(\Succ\).
By an n-clause we understand a clause of the form \(\forall \vec{x} \, ( N(\Param, \vec{x}) \vee C(\vec{x}))\), where \(N(\Param, \vec{x})\) is a disjunction of literals of the form \(\Param \neq t(\vec{x})\) and \(C\) is a disjunction of literals of the form \(t \bowtie s\) where \(\bowtie \ \in \{ =, \neq\}\), and \(t, s\) are terms of sort other than \(\sortNat\).
The formula \(N\) is called the constraint part of the n-clause.
An n-clause set is a conjunction of n-clauses.
For the sake of readability we will sometimes identify an n-clause set with the set of its conjuncts.
The notion of ``cycles'' of the n-clause calculus is based on the descent operator \(\downarrow_{j}\) with \(j \in \mathbb{N}\).
\begin{defi}
  \label{def:1}
  Let \(i \in \mathbb{N}\), and \(\mathcal{C} = \forall \vec{x} \, (N(\Param, \vec{x}) \vee C(\vec{x}))\) an n-clause with \(N = \bigvee_{j = 1}^{k} \Param \neq t_{j}\). Then we define \(\mathcal{C}{\downarrow_{i}} \coloneqq \forall \vec{x} \, (N(\Param,\vec{x}){\downarrow_{i}} \vee C(\vec{x}))\) where \(N{\downarrow_{i}} \coloneqq \bigvee_{j = 1}^{k}\Param \neq \Succ^{i}(t_{j})\).
  For an n-clause set \(S = \bigwedge_{j = 1}^{m}\mathcal{C}_{j}\) we define \(S{\downarrow_{i}} \coloneqq \bigwedge_{j = 1}^{m}\mathcal{C}_{j}{\downarrow_{i}}\).
\end{defi}
Intuitively, the \(\downarrow_{j}\) operation allows us to express that \(\Param\) is replaced by its \(j\)-th predecessor.
The following lemma states a crucial property of the \(\downarrow_{j}\) operator.
\begin{lem}
  \label{lem:2}
  \label{lemma:downarrow_and_inflation_cancel_out_each_other}
  Let \(S(\Param)\) be clause set and \(j \geq 0\), then we have \(S{\downarrow_{j}}(\Succ^{j}\Param) \vdash S(\Param)\).
\end{lem}
The converse of the above entailment does not hold.
However it holds in a theory that provides at least the injectivity of the successor function.
\begin{lem}
  \label{lem:3}
  Let \(S(\Param)\) be a clause set and \(j \geq 0\), then
  \[\forall x \, \forall y \, (\Succ x = \Succ y \rightarrow x = y), S(\Param) \vdash S{\downarrow_{j}}(\Succ^{j}\Param).\]
\end{lem}
We can now introduce the notions of cycle and of refutability by a cycle.
For the sake of brevity we consider a simplified variant of the n-clause calculus defined in \cite{kersani2013combining}.
Only one of the simplifications imposed by us restricts the power of the formalism.
The cycles presented in \cite{kersani2013combining} rely on a decidable entailment relation \(\sqsupseteq\) between clauses such that \(\mathcal{C} \sqsupseteq \mathcal{D}\) implies \(\mathcal{C} \models_{{\mathrm{KP}}} \mathcal{D}\), where \(\models_{\mathrm{KP}}\) is the entailment for the ``Kersani-Peltier'' standard semantics as presented in \cite{kersani2013combining}.
Whenever the original formalism requires \(\mathcal{C} \sqsupseteq \mathcal{D}\), we require that \(\mathcal{C} \models \mathcal{D}\).
On the one hand the relation \(\models\) is not decidable, but on the other hand the relation \(\models\) allows us to use the completeness of first-order logic.
The latter would not be possible with \(\sqsupseteq\) because this relation relies on standard semantics for which the completeness theorem does not hold.
According to \cite{kersani2013combining} the relation \(\sqsupseteq\) is intended to abstract decidable relations such as syntactic equality or subsumption that also satisfy our stronger requirement.
Hence our restriction does not rule out any practically relevant instance of the n-clause calculus.
Finally, our restriction does not limit the generality of Corollary \ref{pro:2} below, since a similar argument could be used, assuming a suitable choice of \(\sqsupseteq\).
\begin{defi}
  \label{def:5}
  Let \(R(\Param)\) be an n-clause set.
  A triple \((i,j, S(\Param))\) with \(i, j \in \mathbb{N}\), \(j >0\) and \(S \subseteq R\) is a cycle for \(R\) if \(S \vdash \Param \neq \numeral{k}\) for \(k = i, \dots, i + j - 1\) and \(S \vdash S{\downarrow_{j}}\).
  We say that \(R\) is refuted by a cycle if there exists a cycle \((i,j,S)\) for \(R\) and \(R \vdash \Param \neq \numeral{k}\), for \(k = 0, \dots, i-1\).
\end{defi}
A cycle \((i,j,S(\Param))\) for a clause set \(R(\Param)\) is similar to an argument by induction with an offset \(i\) and a step with \(j\).
Accordingly, the conditions \(S \vdash n \neq \numeral{k}\) for \(k = i, \dots, i + j - 1\) correspond to the \(j\) base cases, whereas the condition \(S \vdash S{\downarrow_{j}}\) corresponds to the step case.

Cycles of the n-clause calculus are thus structurally similar to clause set cycles.
As announced in Section \ref{sec:org62947a1} we will show that clause set cycles are an abstraction of the cycles of the n-clause calculus.
In order to show that every n-clause set refutable by a cycle is also refuted by a clause set cycle it essentially remains to show that the argument by induction with offset \(i\) and step \(j\) can be turned into an argument by structural induction.
\begin{prop}
  \label{pro:1}
  Let \(R\) be an n-clause set refuted by a cycle, then \(R\) is refuted by a clause set cycle.
\end{prop}
\begin{proof}
  Let \((i,j,S(\Param))\) be a cycle refuting \(R\).
  Consider the formula
\[
  T(\Param) \coloneqq \bigvee_{k = 0}^{j - 1}S(\Succ^{k+i}\Param).
\]
It is not difficult to see that \(T(\Param)\) is logically equivalent to a clause.
Since \(S\) is a cycle, we have \(S(\Param) \vdash \Param \neq \numeral{i + k}\) for \(k = 0, \dots, j - 1\).
Therefore by instantiating \(\Param\) by \(\numeral{i + k}\) we obtain \(S(\numeral{i + k}) \vdash \bot\) for \(k = 0, \dots, j - 1\).
Hence we have \(T(\Zero) \vdash \bot \).

Let \(k \in \{ 0, \dots, j- 2\}\), then we clearly have \(S(\Succ^{k + i + 1}\Param) \vdash T(\Param)\).
Now let \(k = j - 1\).
Since \(S\) is a cycle, we have \(S \vdash S{\downarrow_{j}}\).
Thus by Lemma \ref{lemma:downarrow_and_inflation_cancel_out_each_other} we obtain \(S(\Succ^{i + j}\Param) \vdash S(\Succ^{i}\Param)\).
Therefore \(T(\Succ\Param) \vdash T(\Param)\).
Thus, by the soundness of first-order logic \(T(\Param)\) is a clause set cycle.
Since \(S \subseteq R\) we have \(R \vdash S\) and therefore \(R(\Succ^{i}\Param) \vdash T(\Param)\).
Now let \(k = 0, \dots, i - 1\), then since \(R \vdash n \neq \numeral{k}\) we have \(R(\numeral{k}) \vdash \bot\).
Therefore, by the soundness of first-order logic \(R\) is refuted by the clause set cycle \(T\).
\end{proof}
By the above proposition the notion of refutation of an n-clause set by a cycle is also not stronger than \(\E{1}\) induction.
In the following we will show that an analogue of Theorem \ref{theorem:clause_set_cycles_not_contained_in_open_induction} also holds for the n-clause calculus.
Let \(L\) be the language consisting of the two sorts \(\sortNat\), \(\iota\), the function symbols \(\Zero : \sortNat\), \(\Succ : \sortNat \to \sortNat\), \(+: \sortNat \times \iota \to \iota\), \(c : \iota\), and the predicate symbol \(\TriSymbol : \sortNat \times \iota\).
We will again use the predicate symbol \(\TriSymbol\) as an infix symbol.
Let \(S(\Param)\) be the \(L\) n-clause set consisting of the following n-clauses.
\begin{align}
  \Zero \ & \TriSymbol \ c, \label{clause:n_clause_open_induction_1} \tag{C1} \\
  \forall x \, \forall y \, (\Tri{x}{y} & \rightarrow \Tri{\Succ(x)}{\Succ(x) + y}), \label{clause:n_clause_open_induction_2} \tag{C2} \\
  \forall x \, \forall y \, (\Param = x & \rightarrow \neg \Tri{x}{y}) \label{clause:n_clause_open_induction_3} \tag{C3}.
\end{align}
\begin{lem}
  \label{lem:4}
  The n-clause set \(S(\Param)\) is refuted by a cycle.
\end{lem}
\begin{proof}
  By resolving the clauses \eqref{clause:n_clause_open_induction_1} and \eqref{clause:n_clause_open_induction_3}, we obtain \(S \vdash \Param \neq \Zero\).
  Resolving the clauses \eqref{clause:n_clause_open_induction_2} and \eqref{clause:n_clause_open_induction_3} yields \(\Param = \Succ(x) \rightarrow \neg \Tri{x}{y}\).
  Hence we have \(S \vdash S{\downarrow_{1}}\).
  Thus the triple \((0, 1, S(\Param))\) is a cycle for \(S(\eta)\).
  Therefore \(S(\eta)\) is refuted by a cycle.
\end{proof}
Let us now investigate whether \(S(\eta)\) can be refuted by open induction.
\begin{prop}
  \label{pro:4}
  \(\theoryIndOpen{L} \not \vdash \neg S(\Param)\).
\end{prop}
\begin{proof}
  Assume that \(\theoryIndOpen{L} \vdash \neg S(\Param)\).
  Let \(L'\) be the one-sorted language obtained from \(L\) by replacing the sort \(\iota\) by the sort \(\sortNat\).
 We then have \(\theoryIndOpen{L'} \vdash \neg S(\Param)\).
By replacing the constant \(c\) by \(\Zero\), we obtain
\[
  \theoryIndOpen{L_{\TriSymbol}'} \vdash \neg S[c/\Zero](\Param).
\]
This implies \(\theoryIndOpen{L_{\TriSymbol}'} \vdash \neg S_{\TriSymbol}(\Param)\), thus, contradicting Corollary \ref{cor:1}.
\end{proof}
By Lemma \ref{lem:4} and Proposition \ref{pro:4} we thus have:
\begin{cor}
  \label{pro:2}
  There exists a language \(L\) and an \(L\) n-clause set \(S(\Param)\) refuted by a cycle such that \(\theoryIndOpenL \not \vdash \neg S(\Param)\).
\end{cor}
\section{Conclusion}
We have introduced the concept of clause set cycles and the notion of refutability by a clause set cycle.
Clause set cycles abstract the analogous concepts of cycle and refutability by a cycle of the n-clause calculus.
The main advantage of clause set cycles is their semantic nature, which makes them independent of any inference system.
This independence of an inference system allows for a more general analysis of the properties of this type of cycle.

We have explained clause set cycles in terms of theories of induction.
We first have shown that refutability by a clause set cycle is contained in the theory of \(\E{1}\) induction.
On the other hand refutation by a clause set cycle is not contained in the theory of open induction and we even conjecture that open induction is incomparable with the refutability by a clause set cycle.
Finally, we have transferred these results to the n-clause calculus.
The results allow us to formally situate the strength of a variant of the n-clause calculus with respect to induction, where we formerly only had empirical evidence.
The formal results described in this article improve our understanding of the strength of the approaches for \ac{AITP} based on clause set cycles and help to direct further research.

As mentioned in the introduction, the analysis of clause set cycles is part of a research program which aims at studying methods for automated inductive theorem proving in order to improve the theoretical foundations of this subject.
One of the next questions to consider is how clause set cycles can be extended to handle multiple parameters, how this extension would impact the power of the formalism, and how the addition of parameters can be explained from the perspective of induction.
Another question to consider is how the enhancement of superposition by structural induction presented by Cruanes in \cite{cruanes2017superposition} is related to clause set cycles.
A further topic of interest are approaches to AITP based on cyclic sequent calculi such as the calculus introduced by Brotherston and Simpson in \cite{brotherston2010sequent}.
For instance the inductive theorem prover ``Cyclist'' \cite{brotherston2012generic} is based on the cut-free fragment of this cyclic calculus.
Recently Das \cite{das2020} has shown that in the context of arithmetic the logical consequences of cyclic proofs containing only \(\Sigma_{n}\) formulas are contained in the theory \(I\Sigma_{n + 1}\).
In the setting of arithmetic this result already gives us an upper bound for provers such as ``Cyclist'', however this bound may be improved by taking into account the cut-freeness of the proofs output by ``Cyclist''.

\section*{Acknowledgments}

The authors would like to thank the anonymous reviewers whose insightful feedback has helped to improve this paper.

\bibliographystyle{alpha}
\bibliography{bibliography.bib}


\begin{acronym}
  \acro{AITP}{automated inductive theorem proving}
\end{acronym}

\end{document}